\documentclass[12pt]{article}
\usepackage{amsmath, amsthm}
\newcommand{\ds}{\displaystyle}

\newtheoremstyle{theorem}
  {10pt}		  
  {10pt}  
  {\sl}  
  {\parindent}     
  {\bf}  
  {. }    
  { }    
  {}     
\theoremstyle{theorem}
\newtheorem{theorem}{Theorem}

\newtheoremstyle{defi}
  {10pt}		  
  {10pt}  
  {\rm}  
  {\parindent}     
  {\bf}  
  {. }    
  { }    
  {}     
\theoremstyle{defi}

\begin{document}

\begin{center}
{\Large\bf Symmetry and Singularity Properties of  Steen-Ermakov-Milne-Pinney Equations}\\[5mm]
{\Large K Krishnakumar}\\
{\large Department of Mathematics, Srinivasa Ramanujan Centre, SASTRA Deemed to be University, Kumbakonam 612 001, India.\\krishapril09@gmail.com\\[5mm] A Durga Devi \\ Department of Physics, Srinivasa Ramanujan Centre, SASTRA Deemed to be University, \\Kumbakonam 612 001, India.\\raghadurga@gmail.com\\[5mm] R Sinuvasan \\ Department of Mathematics, VIT-AP University, Amaravathi-522 237, Andhra Pradesh, India.\\rsinuvasan@gmail.com\\[5mm] PGL Leach \\ School of Mathematics, Statistics and Computer Science, University of KwaZulu-Natal,\\ Private Bag X54001\\ and \\ Institute for Systems Science,  \\Department of Mathematics, Durban University of Technology, POB 1334, Durban 4000,\\ Republic of South Africa}\\
 leach@ucy.ac.cy\\[3mm]

{\bf ABSTRACT}
\end{center}

We examine the general element of the class of ordinary differential equations, $yy^{(n+1)}+\alpha y'y^{(n)}=0$, for its Lie point symmetries. We observe that the algebraic properties of this class of equations display an attractive set of patterns, the general member of the class can have three type of Algebra, $(n+1)A_1 \oplus_s\{A_1 \oplus sl(2,R)\}$, $A_1 \oplus sl(2,R)$ or $A_2 \oplus A_1$, for different values of $\alpha$. We look at the singularity properties of these equations for various values of $\alpha$.



\section{Introduction}
The Steen-Ermakov-Milne-Pinney Equation \cite{Steen 74,Ermakov 08,Milne 30,Pinney 50}
$$ \rho''(t)+\omega^2(t)\rho(t)=\frac{1}{\rho(t)^3}$$
is well known for its frequent occurrence in various applications. It is probably  less well known as an integral of a third-order equation of maximal symmetry \cite{Morris 17} and this may well be the source of Pinney's famous solution which was presented without proof. The simplest form  of the equation is obtained by setting $\omega=0$, that is,
$$ \rho''(t)=\frac{1}{\rho(t)^3}.$$
A natural generalisation, namely
\begin{equation}\label{0.1}
yy^{(n+1)}+\alpha y'y^{(n)}=0,
\end{equation}
was studied by Moyo {\it et al.} \cite{Moyo 06} with particular reference to its integrability.

We examine (\ref{0.1}) in terms of its symmetry and singularity properties.

\section{Symmetry Properties}

We consider the ordinary differential equation
\begin{equation}\label{2.1}
yy^{(n+1)}+\alpha y'y^{(n)}=0
\end{equation}
We examine Equation (\ref{2.1}) for its symmetry properties. If $n=1$, there are eight Lie point symmetries given by\footnote{For the calculation of the symmetries we use the Mathematica add-on Sym \cite{Dimas 05 a, Dimas 06 a, Dimas 08 a, Andriopoulos 09 a}.}
\begin{equation}
\begin{array}{llllllll}
\Gamma_1 & = & \partial_x \\
\Gamma_2 & = & x \partial_x\\
\Gamma_3 & = & y \partial_y\\
\Gamma_4 & = & x y \partial_y\\
\Gamma_5 & = & \log(y) \partial_x\\
\Gamma_6 & = & y \log(y) \partial_y\\
\Gamma_7 & = & x^2 \partial_x+x y \log(y) \partial_y\\
\Gamma_8 & = & x \log(y) \partial_x+(y \log(y))^2\partial_y
\end{array}
\end{equation}
for $\alpha=-1.$  Equally for $\alpha \neq -1$ there are eight Lie point symmetries.  Now the value of $\alpha$ intrudes into the expressions for some of the symmetries.  The symmetries are
\begin{equation}
\begin{array}{llllllll}
\Gamma_1 & = & \partial_x \\
\Gamma_2 & = & x \partial_x\\
\Gamma_3 & = & y \partial_y\\
\Gamma_4 & = & y^{-\alpha} \partial_y\\
\Gamma_5 & = & x y^{-\alpha} \partial_y\\
\Gamma_6 & = & \frac{y^{1+\alpha}}{\alpha+1} \partial_x\\
\Gamma_7 & = & (1+\alpha)x^2 \partial_x+x y  \partial_y\\
\Gamma_8 & = & (1+\alpha)x y^{1+\alpha} \partial_x+y^{2+\alpha} \partial_y.
\end{array}
\end{equation}
Because the maximal number of Lie point symmetries for a scalar second-order ordinary differential is eight \cite{Lie 67 a}[p 405], for $n = 1$ the algebra is $sl(3,\,R)$ irrespective of the value of $\alpha$.

When we turn to the second member of the class, namely
\begin{equation} \label{2.2}
y y'''+\alpha y'y'' = 0,
\end{equation}
we obtain three possible algebraic structures. These are
\[ \{ \partial_x, x \partial_x, x^2 \partial_x+2xy \partial_y, y \partial_y, \partial_y, x \partial_y, x^2 \partial_y\}, \]
\[ \{\partial_x, x \partial_x, x^2 \partial_x+x y \partial_y, y \partial_y, \frac{\partial_y}{y},\frac{x \partial_y}{y},\frac{x^2 \partial_y}{y}\} \text{and} \]
\[ \{ \partial_x, x \partial_x, y \partial_y\} \]
corresponding to $\alpha=0$, $\alpha=\frac{3}{1}$ and general $\alpha$, respectively. For $\alpha=0$ and $\alpha=\frac{3}{1}$ the algebra is $3A_1 \oplus_s\{A_1 \oplus sl(2,R)\}$, for general $\alpha$ the algebra is $A_2 \oplus A_1$. (We make use of the Mubarakzyanov Classification Scheme \cite{Morozov 58 a, Mubarakzyanov 63 a, Mubarakzyanov 63 b, Mubarakzyanov 63 c})

The third member of the class is the prototype for subsequent equations and we find the symmetries
\[ \{ \partial_x, x \partial_x, x^2 \partial_x+3xy \partial_y, y \partial_y, \partial_y, x \partial_y, x^2 \partial_y, x^3 \partial_y\}, \]
\[ \{ \partial_x, x \partial_x, x^2 \partial_x+2xy \partial_y,y \partial_y\}\text{and} \]
\[ \{ \partial_x, x \partial_x, y \partial_y\} \]
corresponding to $\alpha=0$, $\alpha=\frac{4}{2}$ and general $\alpha$. For $\alpha=0$  the algebra is $4A_1 \oplus_s\{A_1 \oplus sl(2,R)\}$, for $\alpha=\frac{4}{2}$ the algebra is $A_1 \oplus sl(2,R)$ and for general $\alpha$ the algebra is $A_2 \oplus A_1$.

These observations naturally lead to the following theorem for various values of $\alpha$.

\begin{theorem}
In general, symmetries of an $n^{th}$-order equation of the type $yy^{(n+1)}+\alpha y'y^{(n)}=0 $  are given by
\begin{center}
\begin{tabular}{| c | l |}
  \hline
         &  $ \partial_y, x \partial_y, \cdots x^{(n)}\partial_y $\\
          $\alpha=0$ & $y \partial_y$\\
          & $ \partial_x, x \partial_x+\frac{n}{2} y \partial_y, x^2 \partial_x+nxy \partial_y$\\
  \hline

   $\alpha=\frac{n+1}{n-1}$  &  $\partial_x, x \partial_x+\frac{n-1}{2} y \partial_y,x^2 \partial_x+(n-1) x y \partial_y,y \partial_y$ \\

   \hline

  $\alpha$=else & $\partial_x, x \partial_x,y \partial_y$.\\
  \hline
\end{tabular}
\end{center}

Equation $yy^{(n+1)}+\alpha y'y^{(n)}=0 $ has $n+5$, $4$ and $3$ point symmetries corresponding to $\alpha=0$, $\alpha=\frac{n+1}{n-1}$ and $\alpha=$else, respectively. The corresponding algebras are $(n+1)A_1 \oplus_s\{A_1 \oplus sl(2,R)\}$, $A_1 \oplus sl(2,R)$ and $A_2 \oplus A_1$.

\end{theorem}

\begin{proof}
 One can consider the general equation with left hand side in the form
 \begin{equation}
\Omega: yy^{(n+1)}+\alpha y'y^{(n)}.
\end{equation}
The linearised symmetry condition is $X^{(n+1)} \Omega=0$ when $\Omega = 0$, where
\[
X= \xi(x,y) \partial_x+\eta(x,y) \partial_y
\]
is the generator of the infinitesimal point transformation and
\[
X^{(n)}= \xi(x,y) \partial_x+\eta(x,y) \partial_y+\eta' \partial_{y'}+ \cdots +\eta^{(n)} \partial_{y^{(n)}}
\]
is its extension up to the $n^{th}$ derivative, that is,
\begin{equation} \label{2.3}
-\alpha y'y^{(n)}\eta+\alpha y y^{(n)} \eta^{(1)} +\alpha y y' \eta^{(n)}+y^2\eta^{(n+1)}=0.
\end{equation}
When we use the extension formula
\begin{equation} \label{2.4}
\eta^{(n)}=D^n\eta-\sum_{j=1}^n {n \choose j} y^{(n+1-j)}D^j \xi,
\end{equation}
where $D$ is the total derivative, we can rewrite (\ref{2.3}) as
\begin{equation} \label{2.5}
\begin{split}
-\alpha y'y^{(n)}\eta +\alpha yy^{(n)} (\eta_x+y'(\eta_y-\xi_x)-y'^2\xi_y))+\alpha y y'D^n\eta \\
-\alpha yy' \sum_{j=1}^n {n \choose j} y^{(n+1-j)}D^j \xi+y^2 D^{(n+1)}\eta \\
-y^2 \sum_{j=1}^{n+1} {n+1 \choose j} y^{(n+2-j)}D^j \xi =0.
\end{split}
\end{equation}
Comparing the coefficients of $y''y^{(n)}$ on both sides in equation (\ref{2.5}) we get $\xi _y=0$, that is,
\begin{equation}\label{2.6}
\xi =a(x).
\end{equation}
On substitution of equation (\ref{2.6}) into equation (\ref{2.5}) we have
\begin{equation} \label{2.7}
\begin{split}
-\alpha y'y^{(n)}\eta +\alpha yy^{(n)} (\eta_x+y'(\eta_y-a'))+\alpha y y'D^n\eta \\
-\alpha yy' \sum_{j=1}^n {n \choose j} y^{(n+1-j)}a^{(j)}+y^2 D^{(n+1)}\eta \\
-y^2 \sum_{j=1}^{n+1} {n+1 \choose j} y^{(n+2-j)}a^{(j)}=0.
\end{split}
\end{equation}
By comparison of the coefficients of $y''y^{(n-1)}$ in equation (\ref{2.7}) we see that $\eta_{yy}=0$, that is,
\begin{equation}
\eta=b(x)+ y c(x).
\end{equation}
When we use Leibnitz' rule for differentiating a product, we compute $D^n \eta$ as
\begin{equation}\label{2.8}
D^n \eta=b^n(x)+\ds \sum_{k=0}^n {n \choose k}y^{(k)}c^{(n-k)}.
\end{equation}
On substitution of equation (\ref{2.8}) into equation (\ref{2.7}) we have
\begin{equation} \label{2.9}
\begin{split}
-\alpha (b+y c)y'y^{(n)}+\alpha (b'+y c'+y'(c-a')) y y^{(n)}+\alpha b^{(n)}y y' \\
+\alpha \sum_{k=0}^n {n \choose k} y^{(k)}c^{(n-k)} y y'-\alpha \sum_{j=1}^n {n \choose j} y^{(n+1-j)} a^{(j)} y y' \\
+b^{(n+1)} y^2+\sum_{k=0}^{n+1} {n+1 \choose k} y^{(k)} c^{(n+1-k)}y^2\\
-\sum_{j=1}^{n+1} {n+1 \choose j} y^{(n+2-j)}a^{(j)}y^2=0.
\end{split}
\end{equation}
When we compare the coefficients of $y'y^{(n)}$  in equation (\ref{2.9}), we obtain
\begin{equation} \label{2.10}
b(x)=0.
\end{equation}
Using the equation (\ref{2.10}) we rewrite equation (\ref{2.9}) as
\begin{equation} \label{2.11}
\begin{split}
\alpha c'y^2 y^{(n)}+\alpha \sum_{k=0}^{n-1} {n \choose k} y^{(k)}c^{(n-k)} y y'
-\alpha \sum_{j=2}^n {n \choose j} y^{(n+1-j)} a^{(j)} y y' \\
+\sum_{k=0}^{n} {n+1 \choose k} y^{(k)} c^{(n+1-k)}y^2
-\sum_{j=2}^{n+1} {n+1 \choose j} y^{(n+2-j)}a^{(j)}y^2=0.
\end{split}
\end{equation}
By comparison of the coefficients of $y'y^{(n-1)}$ and $y^{(n)}$ in equation (\ref{2.11}), we obtain
\begin{eqnarray}
2c'-(n-1)a'' & = & 0 \label{2.12}\\
2(n+1+\alpha) c'-n(n+1) a'' & = & 0. \label{2.13}
\end{eqnarray}
Solving the equation (\ref{2.12}) for $c$ we obtain
\begin{equation} \label{2.14}
c=c_1+\frac{n-1}{2}a',
\end{equation}
where $c_1$ is the  constant of integration. On substitution of equation (\ref{2.14}) into equation (\ref{2.13}) we have
\begin{equation}\label{2.15}
(-\alpha +\alpha  n-n-1)a''=0.
\end{equation}
\subsection{Case 1: $\alpha=\frac{n+1}{n-1}$}
The equation (\ref{2.15}) is satisfied for $\alpha = \frac{n+1}{n-1}$. Comparing the coefficients of $y^{(n-1)}$ we get
\begin{equation}\label{2.15}
3c''-(n-1)a'''=0.
\end{equation}
We substitute (\ref{2.14}) into equation (\ref{2.15}) to obtain $a''' = 0$, that is,
\begin{equation}
a=a_1+a_2 x+a_3 x^2.
\end{equation}
The coefficient functions of the symmetries of the case $\alpha = \frac{n+1}{n-1}$ are
\begin{eqnarray}
\xi & = & a_1+a_2 x+a_3 x^2,\\
\eta & = & c_1 y+\frac{n-1}{2} a_2 y +(n-1)a_3 x y.
\end{eqnarray}
\subsection{Case 2: $\alpha =$ else}
If $\alpha = \text{else}$, then $-\alpha +\alpha  n-n-1 \neq 0$. From (\ref{2.15}) we obtain $a''=0$, that is,
\begin{equation}
a=a_1+a_2 x.
\end{equation}
The coefficient functions of the symmetries of the case $\alpha = \text{else}$ are
\begin{eqnarray}
\xi & = & a_1+a_2 x,\\
\eta & = & c_1 y+\frac{n-1}{2} a_2 y.
\end{eqnarray}

\subsection{Case 3: $\alpha = 0$}
The equation is
\begin{equation}\label{2.16}
b^{(n+1)} +\ds \sum_{k=0}^{n+1} {n+1 \choose k} y^{(k)} c^{(n+1-k)}-\ds \sum_{j=1}^{n+1} {n+1 \choose j} y^{(n+2-j)}a^{(j)}=0.
\end{equation}
We collect the constant and $y$ coefficients in equation (\ref{2.16}) and find that
\begin{equation}
b^{(n+1)}(x)=0 \:,\: c^{(n+1)}(x)=0.
\end{equation}
We rewrite equation (\ref{2.16}) as
\begin{equation}\label{2.17}
\ds \sum_{k=1}^{n} {n+1 \choose k} y^{(k)} c^{(n+1-k)}-\ds \sum_{j=2}^{n+1} {n+1 \choose j} y^{(n+2-j)}a^{(j)}=0
\end{equation}
and collect the coefficients of $y^{(n)}$ and $y^{(n-1)}$ in equation (\ref{2.1}) to obtain
\begin{eqnarray}
2 c^{(1)}-n a^{(2)} & = & 0 \label{2.18}\\
3 c^{(2)}-(n-1) a^{(3)} & = & 0. \label{2.19}
\end{eqnarray}
From equation (\ref{2.18}) $c=c_1+\frac{n}{2}a^{(1)}$ and we substitute this into equation  (\ref{2.19}) to obtain $a^{(3)}=0$, that is,
\begin{equation}\label{2.20}
a=a_1+a_2 x+a_3 x^2.
\end{equation}
The coefficient functions of the symmetries of the case $\alpha=0$ are
\begin{eqnarray}
\xi & = & a_1+a_2 x+a_3 x^2 \quad \mbox{\rm and} \nonumber \\
\eta & = & b_1+b_2 x+\cdots +b_{n+1} x^{n}+c_1y+\frac{n}{2}a_2 y+a_3 n x y.\nonumber
\end{eqnarray}

\end{proof}

\section{Singularity Analysis}

We examine the specific class of equations for the value of $\alpha = -(n+1)$ introduced above in terms of singularity analysis. We examine the sequence of equations introduced above in terms of singularity analysis. We follow the general method as outlined in \cite{Ramani 89 a, Tabor 89 a} with the modification for negative nongeneric  resonances introduced by Andriopoulos {\it et al} \cite{Andriopoulos 06 a}. We illustrate the method on the fifth-order equation,
\begin{equation}
yy^{(5)}-5 y' y^{(4)}=0.
\end{equation}
To determine the leading-order  behaviour we set $y=\alpha \chi ^p$, where $\chi=x-x_0$ and $x_0$ is the location of the putative singularity. We obtain
\[
a^2 (p-4) (p-3) (p-2) (p-1) p \chi^{2 p-5}-5 a^2 (p-3) (p-2) (p-1) p^2 \chi^{2 p-5}
\]
which is zero if $(p-4)=5p$, {\it ie}, $p=-1$\footnote{The positive integral values of $p$ sre not acceptable for singularity analysis.}. Note that the coefficient of the leading-order term is arbitrary.

To establish the terms at which the remaining constants of integration occur in the Laurent Expansion we make the substitution
\[
y=\alpha \chi^{-1} +m \chi^{-1+s},
\]
where the various values at which $s$ may take are determined by the coefficient of terms linear in $m$ being zero and so $m$ is arbitrary. The coefficient of $m$ is a fifth-order polynomial in $s$, the roots of which are
\[
s=-1,\,0,\,6,\,\frac{1}{2} \left(5-i \sqrt{39}\right),\,\frac{1}{2} \left(5+i \sqrt{39}\right).
\]
 The resonances are discordant. However, this problem can be overcome by the substitution $y(x) \rightarrow \frac{1}{w(x)}$. From Table1 the value of the leading-order, $-1$, is always present. For a system to possess the Painlev\'e Property the resonance must be an integer. If $\alpha < n$, the Laurent expansion is known as a Right Painlev\'e Series because the exponents commence at $-1$ and increase to a presumed infinity. If $\alpha > n$, the Laurent expansion is known as a Full Painlev\'e Series.

 In  the case  that  the resonance  is  a  rational  number  the expansion can be made in terms of fractional powers -- the same be true if the exponent of the leading-order behaviour be rational. In this case the solution cannot be analytic. Rather, it has branch point  singularities.  Provided  the  denominator  of  the  fractional power is not great, the expansion is acceptable. If the dominator is large, the complex plane is divided by so many branch cuts as to  be  effectively  useless  for  the  almost  inevitable  numerical computations used in the solution. When fractional powers are included in the expansion, the system is said to possess the weak Painlev\'e Property.

 \begin{table}[!htbp]
\caption{Leading order and resonances of Equation (\ref{2.1}) under the transformation $y(x)\rightarrow \frac{1}{w(x)}$}
\begin{center}
\begin{tabular}{|c|c|c|}
\hline
n & Leading-order & Resonances \\
\hline
$2$ & $-1,0,-\frac{2}{1+\alpha}$ & $-1,0,1-\alpha$\\
\hline
$3$ & $-2,-1,0,-\frac{3}{1+\alpha}$ & $-1,0,1,2-\alpha$\\
\hline
$4$ & $-3,-2,-1,0,-\frac{4}{1+\alpha}$ & $-1,0,1,2,3-\alpha$\\
\hline
$n$ & $-n+1, -n+2, \cdots , -1, 0,-\frac{n}{1+\alpha}$ & $-1,0,1,\cdots ,n-2, n-1-\alpha$\\
\hline
\end{tabular}
\end{center}
\label{default}
\end{table}%
Consistency is automatically satisfied as all terms in the equation are dominant. We deduce the following theorem.

 \begin{theorem}
The exponent of the leading-order term and the resonances of the $n$th member of the class of equations,
\begin{equation*}\label{3.1}
yy^{(n+1)}+ \alpha  y' y^{(n)}=0, ~~~~n>1, \alpha ~~\text{rational},
\end{equation*}
under the transformation $y(x) \rightarrow \dfrac{1}{w(x)}$ are $p=-n+1, -n+2, \cdots , -1, 0,-\frac{n}{1+\alpha}$ and $s=-1,0,1,\cdots ,n-2, n-1-\alpha$.
\end{theorem}

\begin{proof}
Substituing $y={w^{-1}}$ in equation (\ref{2.1}) then we have
\begin{equation}
w^{-1}(w^{-1})^{(n+1)}+\alpha(w^{-1})^{\prime}(w^{-1})^{(n)}=0\label{sing1}
\end{equation}
To find the leading order of equation (\ref{sing1}), let us take $w=\beta x^{p}$, Obviously
\begin{equation}
w^{-1}=\frac{1}{\beta}x^{-p}
\end{equation}
Then we may obtained the corresponding first and second order derivative is
\begin{eqnarray}
{(w^{-1})}^{\prime}=\frac{1}{\beta}(-1)p x^{-p-1}\\
{(w^{-1})}^{\prime\prime}=\frac{1}{\beta}(-1)^{2}p(p+1)x^{-p-2}
\end{eqnarray}
we can rewritting the general form as
\begin{eqnarray}
{(w^{-1})}^{(n)}=\frac{1}{\beta}(-1)^{n}p(p+1)(p+2)....(p+(n-1))x^{-p-n}\label{sgen1}\\
{(w^{-1})}^{(n+1)}=\frac{1}{\beta}(-1)^{n+1}p(p+1)(p+2)....(p+(n-1))(p+n)x^{-p-n-1}\label{sgen2}
\end{eqnarray}
By substituting the equation (\ref{sgen1}) and (\ref{sgen2}) to the equation (\ref{sing1})
\begin{eqnarray}
\frac{1}{\beta^{2}}(-1)^{n+1}p(p+1)(p+2)...(p+n-1)(p+n)x^{-2p-n-1}\nonumber&&\\
+\frac{\alpha}{\beta^2}(-1)^{n+1}p^2(p+1)(p+2)...(p+n-1)x^{-2p-n-1}=0
\end{eqnarray}
Collecting the coefficients of $x^{-2p-n-1}$ and equating it into zero
\begin{equation}
p(p+1)(p+2)...(p+n-1)(p+n+\alpha p)=0
\end{equation}
We get the following values for $p$, $p=0,-1,-2,...-n+1, \frac{-n}{1+\alpha}$
To find the resonances let us take $w=\beta x^{-1}+m x^{-1+s}$ therefore
\begin{equation}
w^{-1}=\frac{x}{\beta}(1+\frac{m}{\beta}x^s)^{-1}
\end{equation}
\begin{equation}\label{sing2}
w^{-1}=\frac{1}{\beta}\sum_{k=0}^{\infty}(-1)^k{(\frac{m}{\beta})}^k x^{ks+1}
\end{equation}
\begin{equation}\label{sing3}
{(w^{-1})}^{\prime}=\frac{1}{\beta}\sum_{k=0}^{\infty}(-1)^k{(\frac{m}{\beta})}^k (ks+1)x^{ks}
\end{equation}
\begin{equation}\label{sing4}
{(w^{-1})}^{\prime\prime}=\frac{1}{\beta}\sum_{k=0}^{\infty}(-1)^k{(\frac{m}{\beta})}^k (ks+1)(ks)x^{ks-1}
\end{equation}
\begin{equation}\label{sing5}
{(w^{-1})}^{(n)}=\frac{1}{\beta}\sum_{k=0}^{\infty}(-1)^k{(\frac{m}{\beta})}^k (ks+1)(ks)(ks-1)....(ks-(n-2))x^{ks-(n-1)}
\end{equation}
\begin{equation}\label{sing6}
{(w^{-1})}^{(n+1)}=\frac{1}{\beta}\sum_{k=0}^{\infty}(-1)^k{(\frac{m}{\beta})}^k (ks+1)(ks)(ks-1)....(ks-(n-2))(ks-(n-1))x^{ks-n}
\end{equation}
Substituting equations (\ref{sing2}, \ref{sing3}, \ref{sing5}) and (\ref{sing6}) in equation (\ref{sing1}) then the resultant equation is given by
\begin{eqnarray*}\label{sing6}
\frac{1}{\beta^2}x^{-n+1}\left(\sum_{k=0}^{\infty}\bigg\{\sum_{j=0}^{k}(-1)^k{(\frac{m}{\beta})}^k (js+1)(js)(js-1)....(js-(n-2))(js-(n-1))\bigg\}\right.&&\\\left.-\alpha \sum_{k=0}^{\infty}\bigg\{\sum_{j=0}^{k}(-1)^k{(\frac{m}{\beta})}^k (js+1)(js)(js-1)....(js-(n-2))((k-j)s+1)\bigg\}\right)x^{ks}
\end{eqnarray*}
To collect the coefficient of $m$, we have to take $k=1$ and $j=1$ therefore
\begin{equation}
(s+1)(s)(s-1)....(s-(n-2))(s-n+1+\alpha)=0
\end{equation}

Hence the resonances are $s=-1,0,1,\cdots ,n-2, n-1-\alpha$

\end{proof}

 \section{Discussion}
 We have examined the equation,
 $$ yy^{(n+1)}+\alpha y'y^{(n)}=0, $$
 in terms of the algebraic properties of its Lie point symmetries and its integrability in terms of analytic functions. We found that the Lie point symmetries for general $n$ are
 \begin{description}
 \item
 $\alpha = 0:  \partial_y, x \partial_y, \cdots x^{(n)}\partial_y, y \partial_y, \partial_x, x \partial_x+\frac{n}{2} y \partial_y, x^2 \partial_x+nxy \partial_y$
 \item
 $\alpha =\frac{n+1}{n-1}: \partial_x, x \partial_x+\frac{n-1}{2} y \partial_y,x^2 \partial_x+(n-1) x y \partial_y,y \partial_y$
 \item
 $\alpha = \text{else}: \partial_x, x \partial_x,y \partial_y$
 \end{description}
 The algebras are $(n+1)A_1 \oplus_s\{A_1 \oplus sl(2,R)\}$, $A_1 \oplus sl(2,R)$ and $A_2 \oplus A_1$. respectively. In terms of singularity analysis, under the transformation $y(x)\rightarrow \frac{1}{w(x)}$ the solution for $w(x)$ is either analytic over the complex plane or on a portion of it defined by branch cuts. It follows that $y(x)$ is also analytic.

 \section*{Acknowledgements}

PGLL thanks the University of KwaZulu-Natal and the National Research Foundation of South-Africa for financial support.

\end{document}